\documentclass{article}

\usepackage{latexsym}
\usepackage{amsmath}

\def\C{\mathrm{\,\raise 0.33ex\hbox{\scriptsize\bf(}\!\!\!C}}
\def\prob#1{\mbox{{\rm Prob}}\{#1\}}

\newenvironment{proof}{\textbf{Proof.}}{\hfill \textbf{$\Box$}}
\def\trace{\operatorname{Tr}}

\newcommand{\tbl}[1]{Table~\ref{table.#1}}

\newcommand{\tbllabel}[1]{\label{table.#1}}

\newcommand{\ket}[1]{ {\left| #1 \right\rangle} }
\newcommand{\aket}[1]{ {\left\langle #1 \right|}}
\newcommand{\braket}[2]{ {\langle #1 | #2 \rangle}}  
\newcommand{\braOPket}[3]{ {\langle #1 | #2 | #3 \rangle}}  
\newcommand{\infoGain}{I_{\mbox{\scriptsize gain}}}

\newtheorem{theorem}{Theorem}[section]

\newtheorem{lemma}{Lemma}[section]

\begin{document}

\title{Private Database Queries Using Quantum States with Limited Coherence Times}
\author{Tad Hogg\\{\small Hewlett-Packard Laboratories}\\{\small Palo Alto,CA} \and Li Zhang\\{\small Microsoft Research Silicon Valley}\\{\small Mountain View, CA}}
\maketitle

\begin{abstract}
We describe a method for private database queries using exchange
of quantum states with bits encoded in mutually incompatible
bases. For technology with limited coherence time, the database
vendor can announce the encoding after a suitable delay to allow
the user to privately learn one of two items in the database
without the ability to also definitely infer the second item. This
quantum approach also allows the user to choose to learn other
functions of the items, such as the exclusive-or of their bits,
but not to gain more information than equivalent to learning one
item, on average. This method is especially useful for items
consisting of a few bits by avoiding the substantial overhead of
conventional cryptographic approaches.
\end{abstract}

\section{Introduction}

Quantum information processing~\cite{chuang00} provides potentially
significant performance improvements over conventional techniques.
One example is quantum computation with its ability to rapidly solve
problems, such as factoring~\cite{shor94}, which appear to be
otherwise intractable. However, implementing machines with enough
bits and coherence time to solve computational problems difficult
enough to be of practical interest is a major challenge. Another
application, quantum cryptography, is feasible today for exchanging
keys over distances of tens of kilometers. A third application area
is to quantum economic mechanisms, which can offer benefits with
only a few qubits which should be feasible to implement relatively
soon.

Early quantum information technology is likely to be characterized
by few operations before decoherence, limited ability to store
coherent states and communication involving limited entanglement,
particularly restricted to pairwise entangled states. Such limited
technology will not provide significant computational advantages
over conventional techniques. Nevertheless, limited quantum
information capabilities provide alternative economic mechanisms for
situations benefiting from correlated behaviors among participants
and the information security provided by quantum states. Examples
include provisioning public goods~\cite{chen03,zhang03},
coordinating random choices without
communication~\cite{huberman03,mura03} and
auctions~\cite{harsha06,guha07}. These economic methods can function
with limited quantum information technology, but do not require it.
In contrast, the private database query method described in this
paper relies on the difficulty of maintaining coherence for long
times to be economically viable.

Privacy-enhanced mechanisms can be instrumental in encouraging
beneficial transactions in situations where participants face
economic or other costs if their information is revealed to
others. Examples based on cryptographic methods include allowing
long-term surveys on sensitive social or medical
topics~\cite{huberman02} and auctions~\cite{naor99}. The problem
treated in this paper arises when a user wishes to purchase some
information from a vendor without revealing what information is
desired, but also not paying for additional information (e.g.,
purchasing the entire database).

In the remainder of this paper, we first describe the private
database query problem in its most basic case: selecting one of
two bits. We then prove, by applying the generalized entropic uncertainty relations proven in~\cite{maassen88}, that the user can learn at most one bit under
the assumption that maintaining coherence beyond a limited time is
not possible with the technology available or too costly compared
to the economic value of the information. We then briefly consider
the generalization to larger databases with many bits of
information on multiple items, and conclude with a discussion of
possible applications.

The mechanism provided here differs from providing private
information exchange with cryptographic methods (i.e., learning
exactly one bit and nothing more), or quantum attacks relying on
more advanced quantum technology such as creating and storing
entangled states until completing the protocol~\cite{lo96}. This
illustrates the importance of understanding plausible capabilities
of adversaries, especially in the context of an emerging
technology where advanced capabilities are likely to be too
expensive (or not available) to justify the cost compared to, say,
just purchasing the additional information from the vendor.

The problem we deal with in this paper is known as symmetrically
private information retrieval (SPIR) or oblivious transfer (OT) in
the cryptography community.  We refer to~\cite{gasarch04} for an
excellent survey of the subject.  In the classical computation
model, one can achieve computationally secure SPIR with a single
server under appropriate computational hardness
assumptions~\cite{kushilevitz97,cachin99,naor99-2,mishra00}. When
there are multiple servers which do not communicate with each
other, one can design information theoretically secure
SPIR~\cite{gertner00}. The emphasis in those studies is on
reducing the communication complexity.

Quantum channels allow reducing the communication complexity in the
case of more than one server~\cite{kerenidis03}. The method studied
in this paper only uses one way communication from the vendor to the
user. Therefore, the privacy of the user is guaranteed, and only the
privacy of the vendor is of concern. In such a model, if the length
of the database is $n$ bits, the vendor must communicate $\Omega(n)$
bits of information, and it is impossible to guarantee the vendor's
privacy in the information theoretic sense even with prepared
entangled states between the vendor and the
user~\cite{nayak99,ambainis02}. In this paper, we show that when the
user can only store entangled states for a short time, we may
achieve information theoretic SPIR. This has similar flavor to the
previous study in which the user is memory
constrained~\cite{cachin98}.

\section{One out of two exchange}

In this section, we consider the case when the vendor has a
database of two items, each with $m$ bits, and wishes to deliver
one and only one item to the user according to the user's private
choice, which is not revealed to the vendor.

The vendor picks an encoding for the value of each message,
represented as $2m$ bits. These bits are sent to the user who then
must measure them, in some choice of basis, and wait for the vendor
to announce the choice of encoding. Knowing the measurement outcome
and the encoding allows the user to determine all $m$ bits of one
item. To prevent the user from learning both messages (with
probability 1), it is important that the user's measurement take
place before the vendor announces the encoding -- otherwise the user
could invert the operator producing the announced encoding and
recover both items. The protocol ensures this based on a limited
coherence time of the available technology -- the vendor simply
waits until well past the coherence time before announcing the
encoding choice.

\subsection{Single bit exchange}\label{sec:single:example}

As an illustrative example, consider a database with two items, each
with one bit of information (e.g., a recommendation to buy or sell a
company's stock). The vendor selects two maximally incompatible
measurement bases:
\begin{equation}
\ket{0},\;\ket{1}
\end{equation}
and
\begin{equation}
\frac{1}{\sqrt{2}}(\ket{0}+\ket{1}),\;\frac{1}{\sqrt{2}}(\ket{0}-\ket{1})
\end{equation}
The vendor encodes the database into a superposition of the two bits
such that measurement in these bases reveals the bit value
corresponding to the first or second item in the database,
respectively.

To do this, the vendor encodes the value for each item in two bits,
randomly choosing one of two encodings. For the first encoding, the
vendor specifies the first and second bits using the first and
second of these bases, respectively. The superposition sent to the
user is one of the following, according to whether the database
values are $00$, $01$, $10$ or $11$, respectively:
\begin{eqnarray*}
\frac{1}{\sqrt{2}}(\ket{00}+\ket{01})&&\cr
\frac{1}{\sqrt{2}}(\ket{00}-\ket{01})&&\cr
\frac{1}{\sqrt{2}}(\ket{10}+\ket{11})&&\cr
\frac{1}{\sqrt{2}}(\ket{10}-\ket{11})&&\cr
\end{eqnarray*}

The second encoding specifies the bits using the same bases, but
forms the superposition of the two bits in the opposite order. The
corresponding superpositions for database values $00$, $01$, $10$ or
$11$, respectively, are:
\begin{eqnarray*}
\frac{1}{\sqrt{2}}(\ket{00}+\ket{10})&&\cr
\frac{1}{\sqrt{2}}(\ket{00}-\ket{10})&&\cr
\frac{1}{\sqrt{2}}(\ket{01}+\ket{11})&&\cr
\frac{1}{\sqrt{2}}(\ket{01}-\ket{11})&&\cr
\end{eqnarray*}

The user can obtain the single bit associated with either item by
electing to measure in one of the two announced bases.

For example, suppose the bit values in the database are $01$. If the
vendor chooses the first encoding, the user receives the state
$\ket{\psi}=(\ket{00}-\ket{01})/\sqrt{2}$ as represented in the
first basis. Measurement in the first basis gives either $00$ or
$01$, and the user will use this measured result to learn the value
for the first item is $0$ once the vendor announces the choice of
encoding. Expressed in the second basis, the state $\ket{\psi}$ is
$(\ket{01}+\ket{11})/\sqrt{2}$. Thus choosing to measure in the
second basis gives the user either $01$ or $11$, which specifies the
value of the second item is $1$.

On the other hand, if the vendor chooses the second encoding, the
user receives the state $\ket{\phi}=(\ket{00}-\ket{10})/\sqrt{2}$.
Measurement in the first basis gives either $00$ or $10$, and the
user will learn the value for the first item is $0$ (based on the
2nd bit of either of these outcomes) once the vendor announces the
choice of encoding. Expressed in the second basis, the state
$\ket{\phi}$ is $(\ket{10}+\ket{11})/\sqrt{2}$. Thus choosing to
measure in the second basis gives either $10$ or $11$, indicating
the value of the second item is $1$ (based on the 1st bit of either
of these outcomes).

In either of these cases, the measured outcome in one basis gives
no information on the value of the item associated with the other
basis.

The user could instead choose any other basis for the measurement,
or even use different bases for the two bits. Such choices can
reveal a function of the bit values for both items, e.g., their
exclusive-or. That is, the user could learn whether both bits in the
database are the same without learning anything about their values
(e.g., which could be used as a recommendation regarding a joint
derivative instrument on both companies, such as a recommendation of
whether they will both change in the same direction). However, as
described below, no choice of basis allows the user to learn more
than one bit of information about the database, on average. Thus in
contrast to conventional (i.e., cryptographic) methods for private
database query, the user has a wider set of options than just
picking one of the database bits to learn.

If the user knows the encoding used by the vendor, then instead of
selecting a basis to measure the states, the user could apply the
inverse of the 2-qubit encoding operation to produce values for both
bits in the database upon final measurement.  However, doing this
using the wrong encoding gives no information about either bit. Thus
with limited coherence time, the vendor simply waits longer than
that interval before announcing the choice of encoding. In that case
the user must make a measurement \emph{before} learning the
encoding. After learning the encoding, the user would then know
whether both or neither of the bits are revealed, but would no
longer have the original quantum state. Each alternative occurs with
probability $1/2$ so, on average, only one bit of information from
the database is revealed. The question now is whether there exists
some basis for the user to learn more than one bit of information.
We show this is impossible in the following sections.

\subsection{General formulation}\label{sec:single:general}

Suppose the vendor has two one-bit values, $d_0$ and $d_1$. These
two values specify the vendor's database state $d=\ket{d_0d_1}$. The
vendor has two encoding operators acting on the two bits of the
database, specified as $4\times 4$ encoding matrices $E_0$ and
$E_1$. These matrices are unitary, i.e., $E_i^\dagger = E_i^{-1}$
where $E_i^\dagger$ denotes the adjoint (i.e., complex conjugate
transpose) of $E_i$. The vendor announces these two operators to the
user.

The vendor then randomly picks one of the two encodings, say $E_i$,
and sends the state $e = E_i d$ to the user. That is, $e$ is column
$d$ of the encoding matrix $E_i$. The user selects a measurement
basis, given as the columns of a unitary matrix $M$. This
measurement is a standard projective or von-Neumann measurement. We
consider the more general POVM case in Section~\ref{sec:POVM}. Thus
the user measures the state $M e$, obtaining outcome $j$ with
probability $P(j|d,i)=|(M e)_j|^2 = |(M E_i)_{j,d}|^2$, conditioned
on the vendor's choice of encoding $i$ and the value of the database
$d$. The user is free to choose any basis, but selects $M$ and
performs the measurement without knowledge of which encoding the
vendor selected. After the measurement, the vendor announces the
encoding choice $i$.

From the observation $j$ and choice of encoding $i$ the user can use
Bayes' theorem to obtain a posterior probability distribution over
the values of the database. Specifically,
\begin{equation}\label{eq:Bayes}
P(d|j,i) = \frac{P(j|d,i) P(d)}{\sum_{d'} P(j|d',i) P(d')}
    = \frac{P(j|d,i)}{\sum_{d'} P(j|d',i)}
\end{equation}
where $P(d)$ is the prior distribution on the database, and the last
expression follows from our assumption this prior is uniform, i.e.,
$P(d)=1/4$ is independent of $d$. The sum in the denominator is
$\sum_{d'}|(M E_i)_{j,d'}|^2 =1 $ since the matrix $M E_i$ is
unitary. Thus the user's knowledge of the database, given by the
distribution on the values $d$, is $P(d|j,i) =|(M E_i)_{j,d}|^2$.

The uncertainty in the user's knowledge of the database after this
procedure is the entropy of this distribution, $h_{j,i} =
H(\{P(0|j,i),\ldots,P(3|j,i)\})$ where for a probability
distribution $P=\{p_0,\dots,p_{n-1}\}$, $H(P)=-\sum_k p_k\log p_k$
is its entropy\footnote{Throughout the paper, all the logarithms are
base $2$ unless explicitly stated.}.

Since each encoding choice is equally likely and not known to the
user at the time of measurement, on average the user's remaining
uncertainty about the database is $h_j = (h_{j,0}+h_{j,1})/2$. The
entropy of the prior distribution, $H(1/4,1/4,1/4,1/4)=2$ so the
amount of information the user gains, averaged over the vendor's
choice of encoding, is $2-h_j$. That is, this is the expected value
of the reduction in the user's uncertainty (i.e., entropy) of the
inferred distribution over the database items when the vendor
chooses each encoding with equal probability.

To bound the user's information gain, we need a lower bound on
$h_j$. We obtain such a bound as a special case of the generalized
entropic uncertainty relations~\cite{maassen88}. For a complex unit
vector $u=(u_0,\dots,u_{n-1})^t\in \C^n$, define
$H_2(u)=H(|u_0|^2,\cdots,|u_{n-1}|^2)=-\sum_i p_i\log p_i$ where
$p_i=|u_i|^2$.  For any two $n\times n$ matrices $A$ and $B$, let
$c(A,B)=L_\infty(A B)$, where $L_\infty(X)$ is the infinity-norm for
a matrix $X$, i.e., $L_\infty(X)=\max_{i,j} |X_{ij}|$.
In~\cite{maassen88}, it is shown that
\begin{theorem}\label{thm:entropy}\cite{maassen88}
For any vector $u$ and unitary matrices $A$ and $B$,
$$H_2(Au)+H_2(Bu)\geq -2\log c(A,B)\,.$$
\end{theorem}
In particular, suppose $A B$ is a Hadamard matrix --- an $n\times n$
matrix $W=(w_{ij})$ is a Hadamard matrix if $W$ is unitary and if
$|w_{ij}|=1/\sqrt{n}$. Then $c(A,B) = 1/\sqrt{n}$ and
$H_2(Au)+H_2(Bu)\geq \log n$.

To apply this bound to our case, let $u^\dagger=\hat{e}_j^\dagger M
E_0$ and $v^\dagger=\hat{e}_j^\dagger M E_1$ where $\hat{e}_j$ is
the unit vector with $(\hat{e}_j)_j=1$ and $(\hat{e}_j)_k=0$ for
$k\neq j$. Then $u_d^\dagger = (M E_0)_{j,d}$ so $P(d|j,0)=|u_d|^2$
and, similarly, $P(d|j,1)=|v_d|^2$. We have $v = E_1^\dagger E_0 u$,
so taking $A$ to be the identity matrix and $B=E_1^\dagger E_0$ we
have
$$H_2(Au)+H_2(Bu) = H_2(u) + H_2(v) = h_{j,0}+h_{j,1}$$
and so $h_j \geq -\log c(I,E_1^\dagger E_0)$ from
Theorem~\ref{thm:entropy}. This bound is independent of $j$, i.e.,
the particular outcome the user measures.

Thus if the vendor's choice of encodings have $E_1^\dagger E_0$ a
Hadamard matrix, and the user starts with a uniform prior
distribution for the $n$ states, with $H=\log n$, the average
information gain for the user is
\begin{equation}\label{eq:gain}
\infoGain = \log n - h_j  \leq \log n -\log \frac{1}{\sqrt{n}} =
\frac{1}{2}\log n
\end{equation}

The encoding matrices corresponding to the example in
Section~\ref{sec:single:example} are $E_0$, $E_1$ equal to
\[
\begin{array}{ccc}
\begin{pmatrix}
\frac{1}{\sqrt{2}}&\frac{1}{\sqrt{2}}&0&0\cr
\frac{1}{\sqrt{2}}&-\frac{1}{\sqrt{2}}&0&0\cr
0&0&\frac{1}{\sqrt{2}}&\frac{1}{\sqrt{2}}\cr
0&0&\frac{1}{\sqrt{2}}&-\frac{1}{\sqrt{2}}
\end{pmatrix}&&
\begin{pmatrix}
\frac{1}{\sqrt{2}}&\frac{1}{\sqrt{2}}&0&0\cr
0&0&\frac{1}{\sqrt{2}}&\frac{1}{\sqrt{2}}\cr
\frac{1}{\sqrt{2}}&-\frac{1}{\sqrt{2}}&0&0\cr
0&0&\frac{1}{\sqrt{2}}&-\frac{1}{\sqrt{2}}
\end{pmatrix}
\end{array}\]
respectively. The encoding $E_1$ includes the permutation of the
bits of the database.
$E_1^\dagger E_0$ has all entries equal to $\pm 1/2$, i.e., is a
$4\times 4$ Hadamard matrix. Thus, from Eq.~\ref{eq:gain} with
$n=4$, the average amount of information the user gains is at most
one bit. Moreover, this bound is tight: the measurement examples
described in Section~\ref{sec:single:example} provide the user one
bit of information.

Note the bound applies to the \emph{expected} information gain,
averaged over the choice of encodings. As described in
Section~\ref{sec:single:example}, the user could choose to invert
one of the encodings, e.g., take $M=E_0^\dagger$. If the user
correctly guesses the encoding used by the vendor, this procedure
gives the user both bits of the database, i.e., $h_{j,0}=0$. But if
instead the vendor selected $E_1$, the measured outcome is
completely uninformative, with $h_{j,1}=2$. So the average of these
equally likely possibilities, $h_j = 1$, satisfies the bound.

Thus we establish that a vendor, using maximally incompatible bases
for encoding bits, can arrange for the user to learn no more than
one bit of information, on average, about the two bits in the
database.

\subsection{Multiple-bit exchange}

Theorem~\ref{thm:entropy} applies to vectors with any number of
components $n$, not just $n=4$ as used for the single-bit database
items illustrated in Section~\ref{sec:single:example}. It is well
known that for any $m>0$, there exist $2^m\times 2^m$ Hadamard
matrices. The standard construction is
$W_{2^m}=\underbrace{W_2\otimes W_2\otimes \cdots \otimes W_2}_m$,
where $W_2=\frac{1}{\sqrt{2}}\begin{pmatrix}1&1\cr -1 &
1\end{pmatrix}$.  Hence, we can generalize the above scheme to
exchanging two items, each with $m$ bits.

Suppose that the vendor has two messages $d_0,d_1$, each with $m$
bits. Let $\ell=2^m$. Pick unitary $\ell\times\ell$ matrices $A_0$
and $A_1$ such that $A_0^\dagger A_1$ is Hadamard. Thus $A_1^\dagger
A_0$ is also Hadamard. Let $C_0 = A_0 \otimes A_1$, $C_1 = A_1
\otimes A_0$ and $P$ be the permutation matrix to reverse the order
of the items, i.e., mapping $d_0 d_1$ to $d_1 d_0$. The vendor's
encoding operators are then $E_0=C_0$ and $E_1=C_1 P$. That is, for
the first encoding, the vendor sends the column of $C_0$ indexed by
$d_0 d_1$ and for the second encoding sends the column of $C_1$
indexed by $d_1 d_0$, as in the single bit example. With this
formulation, the discussion of Section~\ref{sec:single:general}
applies directly to this case.

Let $M_j=A_j^\dagger\otimes A_j^\dagger$ for $j=0,1$. It is easy to
verify that $M_j C_i$ is the tensor product of two $\ell\times\ell$
matrices with the $((j-i)\mod 2)$-th component being the identity
matrix. This is coincident with the position of item $d_j$ in the
permuted string when the vendor uses the encoding matrix $C_i$.
Therefore, if the user applies the measurement $M_j$ and later
receives the value of $i$, he learns $d_j$ with probability $1$,
regardless of which $C_i$ the vendor uses.

We observe that
\[E_1^\dagger E_0 = P^\dagger C_1^\dagger C_0 = P^\dagger \left( (A_1^\dagger A_0) \otimes (A_0^\dagger A_1) \right)\,,\] which is a
$n \times n$ Hadamard matrix where $n=2^{2m}=\ell^2$ is the number
of possible configurations for the database. Thus, by
Theorem~\ref{thm:entropy}, we have that no matter which measurement
the user applies, if the prior distribution of the messages is
uniform, then the user can learn at most $\log \ell=m$ bits in
expectation.

This construction extends the result of
Section~\ref{sec:single:general} to show a vendor can arrange for
the user to learn no more than $m$ bits of information, on average,
about the two $m$-bit items in the database.

To make it easy to encode and measure, we can let $A_0,A_1$ take the
form of tensor product of $m$ $2\times 2$ matrices so that the
encoding and measurement can be done by single bit operations. One
choice with this property is $A_0=I$ and $A_1=W_\ell$.

There are two drawbacks in the above scheme.  First, although in
expectation the user learns $m$ bits, his chance of learning both
messages is $1/2$ by guessing right which encoding the vendor uses.
One way to reduce this probability is to split $d_0 = d_{11}\oplus
d_{12}\oplus\dots\oplus d_{1r}$ and $d_1 = d_{21}\oplus
d_{22}\oplus\dots\oplus d_{2r}$, where $\oplus$ represents bitwise
xor.  Then the vendor applies the above scheme to the pairs
$(d_{1i},d_{2i})$ for $1\leq i\leq r$.  The honest user can still
learn the message according to his choice.  But if the user wants to
learn both messages, he will have to guess right for each $1\leq
i\leq r$ which happens with probability $1/2^r$.

The other property of our scheme is that the user may choose to
learn any $m$ bits in the combined message $c = d_0 d_1$, rather
than just all $m$ bits of one of the two items.  This would be
useful if the vendor would like to give the user the freedom to
decide which bits to learn.  On the other hand, this property may be
considered as a violation of security, for example, when the vendor
would like the user to learn only one item but not even partial
information about the other item.  This is of course impossible to
achieve in the information theoretic sense.  But the following
simple scheme may prevent the user from learning individual bits of
the original message.  In the modified scheme, the vendor treats
each message as an element in Galois field $GF(2^m)$ and picks two
random $a\neq 0,b\in GF(2^m)$ and applies the above protocol to the
messages $d_0'=a\cdot d_0+b$ and $d_1'=a\cdot d_1+b$, where $\cdot,
+$ are the arithmetics in $GF(2^m)$.  At the end, the vendor
announces $a,b$ together with the encoding scheme he uses. Clearly,
an honest user is still able to recover the message according to his
choice. On the other hand, if the user only learns, for example, a
constant fraction of the message $d_0'$ before he knows $a,b$,
intuitively, it is unlikely that the user can determine any
individual bit of $d_0$.

\subsection{Generalized measurements}\label{sec:POVM}

Our discussion considered users making conventional projective
measurements on the states they receive from the vendor. A more
general possibility is positive operator valued measurements
(POVM)~\cite{chuang00}. In some cases, such measurements can
distinguish quantum states with higher probability than any
projective measurement. In this section, we briefly describe these
measurements and show they provide no benefit in the context of the
two-item database described above.

A POVM consists of a set of $N$ operators $\{R_1,\ldots,R_N\}$ in an
$n$-dimensional Hilbert space. The operators are not necessarily
Hermitian, orthogonal or invertible, and $N$ may be larger than $n$.
These operators satisfy
\begin{itemize}
\item completeness
\begin{equation}
\sum_j R_j^\dagger R_j = I
\end{equation}

\item nonnegativity: for every vector $x$
\begin{equation}
x^\dagger R_j^\dagger R_j x \geq 0
\end{equation}
\end{itemize}

For a system in a pure state $\ket{\psi}$, measurement gives one of
the outcomes $j=1,\ldots,N$, with probability for outcome $j$ equal
to $P(j)=\braOPket{\psi}{R_j^\dagger R_j}{\psi}$. This probability
is nonnegative due to the nonnegativity condition, and $\sum_j
P(j)=1$ due to the completeness condition. The state after
measurement is
\begin{displaymath}
\frac{1}{\sqrt{P(j)}} \; R_j \ket{\psi}
\end{displaymath}

\paragraph{Example.} For a projection measurement, the POVM
consists of $N=n$ orthogonal projection operators: $R_j =
\ket{e_j}\aket{e_j} = e_j e_j^\dagger$ where $e_j$ is the $j^{th}$
unit basis vector of the measurement. The probability to observe
outcome $j$ for state $\ket{\psi}$ is
$\braket{\psi}{e_j}\braket{e_j}{\psi}$ or $|\braket{e_j}{\psi}|^2$,
in which case the resulting state is
$\ket{e_j}\braket{e_j}{\psi}/|\braket{e_j}{\psi}|$, which is
$\ket{e_j}$ up to a phase factor.

In our context of two database items, each with $m$ bits, the
Hilbert space has dimension $n=2^{2m}$. The vendor picks one of two
encoding operators, $E_0$ or $E_1$, and sends the state $e=E_i d$ to
the user. Suppose the user applies the POVM $\{R_j\}$. With
probability $P(j|d,i) = (E_i^\dagger R_j^\dagger R_j E_i)_{dd}$ the
user observes outcome $j$ conditioned on the vendor's encoding
choice $i$ and value of the database $d$.

Eq.~\ref{eq:Bayes} gives user's inference of the distribution of
database values from the outcome of the measurement and the vendor's
announced choice of encoding. The sum in the denominator of
Eq.~\ref{eq:Bayes}, $\sum_{d'} P(j|d',i)$, is $\trace(E_i^\dagger
R_j^\dagger R_j E_i)$ which equals $\trace(R_j^\dagger R_j)$ since
$E_i$ is unitary. In particular, this sum is independent of the
vendor's encoding choice $i$. We denote this sum by $s^2$ with
$s>0$. Eq.~\ref{eq:Bayes} then gives
$$P(d|j,i) = \frac{P(j|d,i)}{s^2} = (E_i^\dagger S^\dagger S
E_i)_{dd}$$ where $S=R_j/s$ so $\trace(S^\dagger S)=1$.

Let $A=S E_0$. Then $S E_1 = A U$ where $U=E_0^\dagger E_1$ is a
Hadamard matrix for the choice of encodings described above. Thus
$P(d|j,0)=(A^\dagger A)_{dd}$ and $P(d|j,1)=(U^\dagger A^\dagger A
U)_{dd}$.

By orthogonalizing $A$, we have $A^\dagger A=\sum_{r=1}^n \lambda_r
v_r v_r^\dagger$ where $\lambda_r\geq 0$, $\sum_r \lambda_r=1$, and
the $v_r$'s are mutually orthogonal unit vectors. For each $r$ we
define two probability distributions over the values $d$ in the
database: $p^{(r)}_d = |(v_r)_d|^2$ and $q^{(r)}_d = |(U^\dagger
v_r)_d|^2$. Then $P(d|j,0)=\sum_r \lambda_r p^{(r)}_d$ and
$P(d|j,1)=\sum_r \lambda_r q^{(r)}_d$.

By Theorem~\ref{thm:entropy}, we have $H(p^{(r)})+H(q^{(r)})\geq
\log n$ because $U$ is a Hadmard matrix. Using this bound and the
convexity of the entropy function, we have
\begin{eqnarray*}
H(\{P(d|j,0)\})+H(\{P(d|j,1)\}) & = & H\left(\sum_r \lambda_r
p^{(r)}\right)+H\left(\sum_r \lambda_r
q^{(r)}\right)\\
& \geq & \sum_r \lambda_r \left(H(p^{(r)})+H(q^{(r)}) \right)\\
& \geq & \log n\sum_r \lambda_r = \log n\,.
\end{eqnarray*}
Thus Eq.~\ref{eq:gain} applies to this POVM, giving the same bound
as for the projective measurements considered above. Since
projective measurements can achieve the lower bound, we see POVM
provides no advantage for the user in this context.

\section{One out of $k$ exchange}

The previous section considered a database of two items.  We showed
how the user could privately learn one item, and no more, when the
vendor selects randomly from two encodings related by a Hadamard
matrix and relies on limited coherence time to force the user to make
a measurement before the choice of encoding is announced.

A natural extension is to a database with $k$ items, each consisting
of $m$ bits. The mechanism would then allow the user to privately
learn a limited number of bits, on average, as well as provide an
opportunity to learn all the bits of any one item with probability
one. \tbl{notation} summarizes our notation.

\begin{table}
\begin{center}
\begin{tabular}{cl}
  \hline
  $m$ & number of bits per item \\
  $k$ & number of items in database \\
  $\ell = 2^m$ & number of possible values for each item \\
  $n = 2^{k m}$ & number of possible database configurations \\
  $A_i$ & a $2^m \times 2^m$ matrix for encoding one item \\
  $C_i$ & an $n \times n$ encoding matrix \\
  \hline
\end{tabular}
\end{center}
\caption{\tbllabel{notation}Summary of notation.}
\end{table}

The scheme is similar to the case of $k=2$.  The vendor chooses
$k$ encodings $C_0,\cdots,C_{k-1}$ in the joint space of
$d_0,\cdots, d_{k-1}$, each $m$-bit long.  The vendor chooses one
encoding randomly among those candidate encodings to encode the
items and send to the user.  The user measures the state before
the vendor announces the encoding.  We would like the following
properties hold:
\begin{enumerate}
  \item For each $0\leq i\leq k-1$, there exists a measurement $M_i$
  such that the user learns $d_i$ for sure after the vendor announces
  the encoding.
  \item For any measurement $M$ made before the vendor announces the encoding, including POVM, the user can only learn at most
  $m$ bits when averaged over the vendor's choice of encoding.
\end{enumerate}
As an extension to $k=2$ case, we consider the following encoding
scheme.  Let $A_0, A_1,\dots,A_{k-1}$ be $\ell\times\ell$ unitary
matrices. Let
\begin{equation}\label{eq:Ci}
C_i = A_i\otimes A_{i+1}\otimes \cdots \otimes A_{k-1}\otimes A_{0}
\cdots \otimes A_{i-1}\,.
\end{equation}

For each $0\leq i\leq k-1$, we use $C_i$ to encode the concatenated
string $c_i = d_i d_{i+1}\dots d_{k-1}d_{0}\dots d_{i-1}$, so the
corresponding encoding matrix is $E_i = C_i P_i$ where $P_i$ is the
permutation giving this reordering of the string. Let
$M_j=A_j^\dagger\otimes A_j^\dagger\cdots\otimes A_j^\dagger$. It is
easy to verify that $M_j C_i$ is the tensor product of $k$
$2^m\times 2^m$ matrices with the $((j-i)\mod k)$-th component being
the identity matrix.  This is coincident with the position of $d_j$
in the concatenated string when using the encoding matrix $C_i$.
Therefore, if the user applies the measurement $M_j$ and later
receives the value of $i$, he learns $d_j$ with probability $1$,
regardless of which $C_i$ the vendor uses.  It is however harder to
guarantee condition 2 which is implied from the following property:
\begin{enumerate}
\item[2'.]
For any unitary vector $u\in \C^n$ where $n=2^{km}$,
\begin{equation}\label{eq:hk}
H_2(C_0 u)+\cdots +H_2(C_{k-1}u)\geq (k-1)\log n\,.
\end{equation}
\end{enumerate}
According to Theorem~\ref{thm:entropy}, when $k=2$, this property is
satisfied by letting $C_0=I$ and $C_1=W_n$.  However, we do not know
the existence of such matrices for $k>2$.

Given the difficulty of finding matrices that satisfy the condition
(2'), we can instead consider the case when the user is honest,
i.e.\ the user always applies one of the measurements $M_0, \ldots,
M_{k-1}$ so to learn one of $d_i$'s for sure. In this case, to
minimize the information leaked, we would like to minimize
$L_\infty(A_j^{\dagger} A_i)$ for $i\neq j$ by
Theorem~\ref{thm:entropy}. For example, if we can construct a set of
$k$ matrices $A_0,\cdots, A_{k-1}$ such that $A_i^\dagger A_j$ is
Hadamard for any $i\neq j$, then the scheme using Eq.~\ref{eq:Ci}
has the honest user learning no information other than the target
item.
Such sets of matrices exist for $k\leq 2^m+1$. Specifically, for any
$m\geq 1$, there exist $2^m+1$ complex matrices $A_i$'s of dimension
$2^m\times 2^m$ such that $A_i^\dagger A_j$ is Hadamard for any
$i\neq j$~\cite{ccks00}.  We refer to~\cite{hsp06} for a simpler
construction.  Therefore, using the construction in~\cite{hsp06} as
the encoding matrices, we can achieve perfect privacy for honest
users as long as $k\leq 2^m+1$.

Furthermore, by Theorem~\ref{thm:entropy}, the encoding of
Eq.~\ref{eq:Ci} with $A_i^\dagger A_j$ Hadamard for any $i\neq j$
leaks at most $km/2=\frac{1}{2}\log n$ bits of information even when
users pick arbitrary measurements because
\begin{eqnarray}\label{eq:mk}
&&H_2(C_0 u)+\cdots +H_2(C_{k-1}u)\nonumber \\
&=&\frac{1}{k-1}\sum_{i\neq j}(H_2(C_i u)+H_2(C_j u))\\
&\geq& \frac{k}{2}\log n\,. \quad\mbox{by Theorem~\ref{thm:entropy}}
\nonumber
\end{eqnarray}
Hence, the encoding leaks at most $\log n-\frac{1}{k}\sum_{i}H_2(C_i
u)\leq \frac{1}{2}\log n$ bits. The convexity argument of
Section~\ref{sec:POVM} applies in this case as well. Thus instead of
the $m$ bits an honest user learns, for general measurements we have
the weaker bound where the user could learn up to $k m/2$ bits.

While we are unable to show any bound other than Eq.~\ref{eq:mk}, by
numerical experiments with small values of $k$ and $m$, we observe
that the number of bits leaked is approximately $0.4k^{0.7}m$ bits.
We note that the lower bound on the sum of entropies
in~\cite{sanchez93}, for $N+1$ complementary observables in
$N$-dimensional Hilbert space, does not apply to our case where the
dimension of the encoding matrices is $2^{km}$, much larger than
$k$, the number of matrices. In addition, the construction
in~\cite{hsp06} is easy to implement physically as they are a
multiplication of a diagonal matrix and the Walsh-Hadamard
transform.

\paragraph{Example.} As a concrete example when $k=3$, consider three matrices $A_0, A_1, A_2$ defined
as follows.  We let $A_0 = I$, $A_1=\alpha_1\otimes \cdots \otimes \alpha_1$,
$A_2=\alpha_2\otimes\cdots\otimes \alpha_2$ where
\[
\begin{array}{ccc}
\alpha_1 = \frac{1}{\sqrt{2}}\begin{pmatrix} 1&1\cr 1&-1
\end{pmatrix}&&
\alpha_2 = \frac{1}{\sqrt{2}}\begin{pmatrix} 1&1\cr -i&i
\end{pmatrix}
\end{array}\]

It can be readily verified that $A_0,A_1,A_2$ satisfy the property
that $A_i^\dagger A_j$ is Hadamard for $i\neq j$.

The construction in~\cite{hsp06} only works when $k\leq 2^m+1$. For
larger $k$'s, we observe that if we simply pick random unitary matrix
according to Haar measure~\cite{katz99}, then each $A_i^\dagger A_j$
is nearly Hadamard by the following lemma.

\begin{lemma}
For two randomly picked $\ell\times \ell$ orthonormal matrices
$A,B$,
\[\prob{L_\infty(A^\dagger B)\geq t}\leq 4\ell^2 e^{-t^2\ell/2}\,.\]
\end{lemma}
\begin{proof}
The proof follows from the measure concentration result on
sphere~\cite{matousek02}. That is, for any two randomly picked unit
vector $u,v$,
\[\prob{|u\cdot v|\geq t}\leq 4 e^{-t^2\ell/2}\,.\]
\end{proof}

Thus, if we pick $k$ random $2^m \times 2^m$ orthonormal matrices
$A_0,\cdots, A_{k-1}$, then
\begin{equation}\label{eq:con}
\prob{\exists i\neq j\,\, L_\infty(A_i^\dagger A_j)\geq t}= O(k^2 2^{2m} e^{-t^2 2^m/2})\,.
\end{equation}

If we let $t\geq c \frac{\log k+m}{2^{m/2}}$ for some constant $c>0$,
then with high probability, $L_\infty(A_i^\dagger A_j)\leq t$ for all
$0\leq i,j\leq k-1$ and $i\neq j$.  Thus, we have the following.
\begin{theorem}\label{thm:entropyk}
For $m=\Omega(\log k)$, if we pick $k$ random orthonormal matrices,
then with high probability, for an honest user, the information he
learns about the other items is $O(k\log(m+\log k))$ bits.
\end{theorem}
\begin{proof}
By Eq.~\ref{eq:con}, the information an honest user learns about any
other item is $m -\log t^2 = O(\log(m+\log k))$ bits. Thus, in total
it is $O(k\log(m+\log k))$ bits.
\end{proof}

One problem with random matrices is that it is hard to realize an
encoding and perform a measurement.  It would be good if each $A_i$
can be further decomposed into the tensor product of smaller
matrices. This can be done when $k$ is much smaller than $m$.  Let
$r=\Omega(2^k m)$.  We pick a random $r\times r$ matrix $B_i$, and
let $A_i$ be the $m/\log r$ tensor product of $B_i$. It can be shown
by a similar argument that an honest user learns $o(m)$ bits of the
other items in addition to the item he chooses. In the case when $k$
is much smaller than $m$, $r$ is much smaller than $2^m$, and it
reduces the complexity of encoding and measurement.

In the above discussion, we consider the natural choice where the
number of encodings used is the same as the number of items.
This, however, does not have to be the case.  More generally, we
can assume that there are $K$ encodings where $K\neq k$.  It is
possible that by using $K>k$, less information is leaked.

To further reduce the complexity, in the above construction, we can
pick a unitary matrix $A$ and let $A_i = A^i$ for $0\leq i\leq k-1$.
This way, the encoding and measurement are simple as there is only
one operator to be implemented.  To reduce the information leaked,
we would like $A$ satisfy the following properties: $A^k = I$ and
$A^i$ is a Hadamard (or nearly Hadamard) matrix for any $1\leq i\leq
k-1$. For $k=3$, we can let $A$ be the tensor product of the
following matrix:
\[
\frac{e^{i \pi/12}}{\sqrt{2}}\begin{pmatrix} 1&1\cr -i&i
\end{pmatrix}
\]

For general $k$'s, again we do not know how to construct such $A$.
It would be interesting to know whether, for any $k$ and
sufficiently large $n$, there exists an $n\times n$ unitary matrix
$A$ so that $A^k=I$ and $A^i$ is Hadamard for $1\leq i\leq k-1$.

When the items contain many bits ($m$ is large), another approach
is to encrypt the items and use the above protocol with the
corresponding decryption keys (which generally use considerably
fewer bits than the size of the encrypted items so the quantum
mechanism would not have to deal with the large number of bits in
each item). In this case, the user picks an item by arranging the
measurement to learn the bits of that item's key. The limit on how
much information the user can learn would then make it difficult
to learn multiple items, provided the keys are long enough. For
example, with the 2-encoding case discussed above, the keys must
be long enough that guessing the remaining bits when half are
known is still not feasible.

\section{Discussion}

In this paper, we described a simple, private database query
protocol using a quantum communication channel. It's ability to
maintain privacy for the vendor relies on an assumption of limited
coherence times for storing and manipulating quantum states. Thus
this protocol is not only suitable for early development of
quantum information technology with limited capability, but
specifically takes advantage of those limitations.

Because the user can choose to learn about combinations of bits of
the database items instead of just a single item for sure, our
protocol presents a larger range of choices for the user than
conventional treatments of oblivious transfer. The extent to which
this would be beneficial depends on the economic context, and
associated incentives, in which the protocol is used. One possible
application is as a component of digital property rights
management. Specifically, the protocol could be useful in
situations where the main economic value is from the combined
inputs of user and vendor, rather than simply the data from the
vendor. That is, private computation of a function of both the
vendor's data and the user's choice as influenced by private
information held by the user. In this case, with a reasonably
large number of items and user choices, even if the user were to
reveal the result to other potential users, that information would
likely have low value to the other users unless they happened to
wish to make the same choice as the original user. Thus those
additional users would also need to purchase the information from
the vendor rather than attempting to free ride on a single user's
purchase.

An interesting direction for future study is generalizing the
protocol to multiple users who have access to additional quantum
channels among themselves. In particular, in some economic
scenarios, users may wish to ensure coordinated choices while still
maintaining as much privacy as possible. In such situations, it
would be useful to identify any benefits of the quantum channel
among users, particularly if limited to pairwise entanglement which
is easier to implement than more general entangled states. The
potential economic benefits of such a protocol should also be
compared with that available with using classical
correlation~\cite{meyer04}. Moreover, in practice, the prior
distribution of database values need not be the uniform distribution
we considered and it would be of interest to evaluate the
consequences of such prior knowledge on the part of the receiver of
the quantum state.

Our quantum mechanism can be simulated classically by having each
player to send a choice of operator to a trusted third party. This
observation, which also applies to other quantum
games~\cite{vanenk02}, means the practical benefit of such quantum
mechanisms depends on the context of the game, e.g., the differences
in security and communication costs as well as the level of trust
assumed for the central institution. For instance, the quantum version
allows only a single measurement of the outcome rather than revealing
the full database, and hence can provide additional privacy for the
vendor. Such privacy can also be achieved via conventional
cryptographic methods but with security based on the apparent
difficulty of solving certain problems, e.g., factoring, rather than
inherent in quantum physics.  In addition, there is a large overhead
when using cryptographic methods, especially
when the number of bits involved is small.

Finally, using game theory to evaluate behavior of economic
mechanisms gives at best approximations of real human behavior. In
this case, rationality dictates that each individual has a full
understanding of the quantum mechanical implications of the
measurement operator choices. How well this describes the actual
behavior of people involved in quantum games is an interesting
direction for future work with laboratory experiments involving
human subjects.

\section*{Acknowledgments}

We thank Kay-Yut Chen, David Fattal, Saikat Guha and Phil Kuekes for
helpful discussions.


\begin{thebibliography}{10}

\bibitem{ambainis02}
Andris Ambainis, Ashwin Nayak, Amnon Ta-Shma, and Umesh~V. Vazirani.
\newblock Dense quantum coding and quantum finite automata.
\newblock {\em J. ACM}, 49(4):496--511, 2002.

\bibitem{cachin98}
Christian Cachin, Claude Cr{\'e}peau, and Julien Marcil.
\newblock Oblivious transfer with a memory-bounded receiver.
\newblock In {\em FOCS}, pages 493--502, 1998.

\bibitem{cachin99}
Christian Cachin, Silvio Micali, and Markus Stadler.
\newblock Computationally private information retrieval with polylogarithmic
  communication.
\newblock In {\em EUROCRYPT}, pages 402--414, 1999.

\bibitem{ccks00}
A.~R. Calderbank, P.~J. Cameron, W.~M. Kantor, and J.~J. Seidel.
\newblock Z4-{Kerdock} codes, orthogonal spreads and extremal {Euclidean}
  line-sets.
\newblock {\em Proceedings of Landon Mathematical Society}, 2000.

\bibitem{chen03}
Kay-Yut Chen, Tad Hogg, and Raymond Beausoleil.
\newblock A quantum treatment of public goods economics.
\newblock {\em Quantum Information Processing}, 1:449--469, 2002.
\newblock arxiv.org preprint quant-ph/0301013.

\bibitem{gasarch04}
William~I. Gasarch.
\newblock A survey on private information retrieval (column: Computational
  complexity).
\newblock {\em Bulletin of the EATCS}, 82:72--107, 2004.

\bibitem{gertner00}
Yael Gertner, Yuval Ishai, Eyal Kushilevitz, and Tal Malkin.
\newblock Protecting data privacy in private information retrieval schemes.
\newblock {\em J. Comput. Syst. Sci.}, 60(3):592--629, 2000.

\bibitem{guha07}
Saikat Guha, Tad Hogg, David Fattal, Timothy Spiller, and Raymond~G.
  Beausoleil.
\newblock Quantum auctions using adiabatic evolution: The corrupt auctioneer
  and circuit implementations.
\newblock {\em International Journal of Quantum Information}, 6:815--839, 2008.

\bibitem{hsp06}
R.~W. Heath, T.~Strohmer, and A.~J. Paulraj.
\newblock On quasi-orthogonal signatures for {CDMA} systems.
\newblock {\em IEEE Transactions on Information Theory}, 52:1217--1226, 2006.

\bibitem{harsha06}
Tad Hogg, Pavithra Harsha, and Kay-Yut Chen.
\newblock Quantum auctions.
\newblock {\em Intl. J. of Quantum Information}, 5:751--780, 2007.
\newblock arxiv.org preprint 0704.0800.

\bibitem{huberman02}
Bernardo~A. Huberman and Tad Hogg.
\newblock Protecting privacy while revealing data.
\newblock {\em Nature Biotechnology}, 20:332, 2002.

\bibitem{huberman03}
Bernardo~A. Huberman and Tad Hogg.
\newblock Quantum solution of coordination problems.
\newblock {\em Quantum Information Processing}, 2:421--432, 2003.
\newblock arxiv.org preprint quant-ph/0306112.

\bibitem{katz99}
Nicholas~M. Katz and Peter Sarnak.
\newblock {\em Random matrices, Frobenius Eigenvalues, and Monodromy}.
\newblock American Mathematical Society, 1999.

\bibitem{kerenidis03}
Iordanis Kerenidis and Ronald de~Wolf.
\newblock Quantum symmetrically-private information retrieval.
\newblock {\em CoRR}, quant-ph/0307076, 2003.

\bibitem{kushilevitz97}
Eyal Kushilevitz and Rafail Ostrovsky.
\newblock Replication is not needed: Single database, computationally-private
  information retrieval.
\newblock In {\em FOCS}, pages 364--373, 1997.

\bibitem{lo96}
Hoi-Kwong Lo.
\newblock Insecurity of quantum secure computations.
\newblock {\em Physical Review A}, 56:1154--1162, 1997.

\bibitem{maassen88}
Hans Maassen and J.~B.~M. Uffink.
\newblock Generalized entropic uncertainty relations.
\newblock {\em Physical Review Letters}, 60:1103--1106, 1988.

\bibitem{matousek02}
Jiri Matousek.
\newblock {\em Lectures on Discrete Geometry}.
\newblock Springer, 2002.

\bibitem{meyer04}
David~A. Meyer.
\newblock Quantum communication in games.
\newblock In S.~M. Barnett et~al., editors, {\em Quantum Communication,
  Measurement and Computing}, volume 734, pages 36--39, Melville, NY, 2004. AIP
  Conference Proceedings, AIP.

\bibitem{mishra00}
S.~Mishra.
\newblock Symmetrically private information retreival.
\newblock Master's thesis, Indian Statistical Institute, 2000.

\bibitem{mura03}
Pierfrancesco~La Mura.
\newblock Correlated equilibria of classical strategic games with quantum
  signals.
\newblock arxiv.org preprint quant-ph/0309033, Sept. 2003.

\bibitem{naor99-2}
Moni Naor and Benny Pinkas.
\newblock Oblivious transfer and polynomial evaluation.
\newblock In {\em STOC}, pages 245--254, 1999.

\bibitem{naor99}
Moni Naor, Benny Pinkas, and Reuben Sumner.
\newblock Privacy perserving auctions and mechanism design.
\newblock In {\em Proc. of the {ACM} Conference on Electronic Commerce (EC99)},
  pages 129--139, NY, 1999. ACM Press.

\bibitem{nayak99}
Ashwin Nayak.
\newblock Optimal lower bounds for quantum automata and random access codes.
\newblock In {\em FOCS}, pages 369--377, 1999.

\bibitem{chuang00}
Michael~A. Nielsen and Isaac~L. Chuang.
\newblock {\em Quantum Computation and Quantum Information}.
\newblock Cambridge Univ. Press, 2000.

\bibitem{sanchez93}
Jorge Sanchez.
\newblock Entropic uncertainty and certainty relations for complementary
  observables.
\newblock {\em Physics Letters A}, 173:233--239, 1988.

\bibitem{shor94}
Peter~W. Shor.
\newblock Algorithms for quantum computation: Discrete logarithms and
  factoring.
\newblock In S.~Goldwasser, editor, {\em Proc. of the 35th Symposium on
  Foundations of Computer Science}, pages 124--134, Los Alamitos, CA, November
  1994. IEEE Press.

\bibitem{vanenk02}
S.~J. van Enk and R.~Pike.
\newblock Classical rules in quantum games.
\newblock {\em Physical Review A}, 66:024306, 2002.

\bibitem{zhang03}
Li~Zhang and Tad Hogg.
\newblock Reduced entanglement for quantum games.
\newblock {\em Intl. J. of Quantum Information}, 1(3):321--335, 2003.

\end{thebibliography}

\end{document}